\documentclass{article}
\usepackage{ijcai26}

\usepackage{times}
\usepackage{soul}
\usepackage{url}
\usepackage{comment}
\usepackage[hidelinks]{hyperref}
\usepackage[utf8]{inputenc}
\usepackage[small]{caption}

\usepackage[color=green!30]{todonotes}

\usepackage{graphicx}
\usepackage{color}
\usepackage{amsmath,amsthm,amssymb,bm}
\usepackage{booktabs}
\usepackage{algorithm2e}
\usepackage{algorithmic}
\usepackage[switch]{lineno}
\newtheorem{definition}{Definition}

\newtheorem{theorem}{Theorem}
\newtheorem{lemma}{Lemma}
\newtheorem{corollary}[theorem]{Corollary}

\makeatletter
\newcommand{\printfnsymbol}[1]{%
  \textsuperscript{\@fnsymbol{#1}}%
}
\makeatother

\title{Resolving Envy by Adding Goods with Bounded Supply: A Type-Count Dichotomy and Two-Agent Hardness\thanks{This work is supported by the National Science and Technology Council under grant no. NSTC (114)-2927-I-032-501 and NSTC 112-2221-E-032-018-MY3.}}

\author{
Chuang-Chieh Lin$^1$ \and 
Guillaume Fertin$^2$ \and 
Po-An Chen$^{3}$ \and 
St\'{e}phane Vialette$^4$ \and
G\'{e}raldine Jean$^2$ \and
Emile Benoist$^2$ \and 
Colin Cleveland$^5$
\affiliations
$^1$National Taiwan Ocean University, Taiwan.\, $^2$Nantes Universit\'{e}, CNRS, LS2N, UMR 6004, F-44000 Nantes,
France.\, 
$^3$National Yang Ming Chiao Tung University, Taiwan.\, $^4$LIGM, Université Gustave Eiffel, CNRS, 77454 Marne-la-Vallée, France.\,
$^5$King's College London\, Bush House, 30 Aldwych, London WC2B 4BG, United Kingdom
\emails
josephcclin@mail.ntou.edu.tw,
\{Guillaume.Fertin, Geraldine.Jean, Emile.Benoist\}@univ-nantes.fr,
poanchen@nycu.edu.tw, 
stephane.vialette@univ-eiffel.fr,
colin.cleveland@kcl.ac.uk
}

\begin{document}
\setlength{\titlebox}{6.5cm}
\maketitle

\begin{abstract}
We study envy elimination by adding goods (EEAG) when the additional pool has
bounded supply and no separate budget bound.  We establish a sharp type-count
dichotomy for binary additive valuations.  With one additional item type, EEAG
is polynomial-time solvable for any number of agents. More generally, our
algorithm permits arbitrary nonnegative integer per-copy values.  The envy
constraints form a system of difference constraints, and Bellman--Ford returns
the componentwise least feasible extension.  In contrast, with exactly two
additional item types, EEAG is \textsf{NP}-complete even when both types have
positive finite supply and the approvers of one type form a subset of the
approvers of the other.  This closes the two-type case left open by Bentert et
al.  Separately, we prove weak \textsf{NP}-completeness even for two agents with
identical additive valuations, one initially endowed good, and a growing number
of unit-supply item types.  Thus, bounded-supply hardness appears both with two
item-types and many agents and with two agents and many item-types.
\end{abstract}

\section{Introduction}
\label{sec:intro}

Fair division studies how resources can be allocated among agents with
heterogeneous preferences while satisfying fairness and efficiency desiderata.
A central fairness notion is envy-freeness: an allocation is envy-free if no
agent prefers another agent's bundle to her own. For divisible resources,
envy-free allocations are often guaranteed to exist, but for indivisible goods
they may fail to exist even in very small instances. This has motivated a large
literature on relaxations such as EF1, EFX, and MMS, as well as on algorithmic
and complexity-theoretic questions surrounding fair allocation of indivisible
goods. See \cite{Amanatidis2023} for a recent survey.

Many real-world allocation problems, however, do not begin from an empty
allocation. Goods may already have been assigned, either because of legal,
administrative, or practical constraints, and the initial allocation may be
unfair. In such settings, reallocating the existing goods may be undesirable or
impossible. Several approaches have been proposed to repair unfair allocations,
including deleting or donating goods, reallocating goods, sharing goods, and
providing monetary subsidies. A recent approach, introduced by Bentert et
al.~\shortcite{BentertEtAl25}, is to eliminate envy by adding goods. In this
framework, one starts from a fixed initial allocation and is allowed to allocate
additional items from a given pool, subject to supply constraints, with the goal
of obtaining an envy-free extended allocation.

Bentert et al.~\shortcite{BentertEtAl25} introduced this problem as
\emph{Envy Elimination by Adding Goods} (EEAG) and studied the problem with  
both unbounded and bounded resources.  If every additional item type has unbounded supply and the
extension has no budget bound, EEAG is polynomial-time solvable.  On the other hand, 
finite supply changes the picture dramatically. 
They prove \textsf{NP}-hardness even with three additional
item types, binary valuations, and an unbounded budget, and obtain further
parameterized hardness and tractability results.  They explicitly ask whether
their three-type hardness result extends to exactly two additional item types.
Thus, our point of departure is not a general bounded-versus-unbounded
dichotomy, but the finer roles of the number of agents and the number of
bounded item types.

In this paper, we resolve the type-count boundary for binary valuations and
also study the orthogonal boundary obtained by fixing the number of agents.
With one additional item type, we allow any number of agents and arbitrary
nonnegative integer per-copy values.  Although the finite supply couples all
assignments through a global cap, the envy constraints are difference
constraints, and their componentwise least feasible solution is computable in
polynomial time.  With exactly two additional item types, we refine the
\textsc{Clique} reduction of Bentert et al.~\shortcite{BentertEtAl25} by merging
their selected-edge and selected-vertex types.  A new counting argument shows
that the merged supply still forces a clique, yielding \textsf{NP}-completeness
and answering their open question.  Separately, with two agents and identical
additive valuations, we prove weak \textsf{NP}-completeness when one agent
initially owns one good and the additional pool contains many types, each
available once.  This last proof is a direct equivalence with \emph{Equal Sum
Subsets with an Enforced Element} (ESSE)~\cite{CieliebakEtAl08}.

\paragraph{Our contributions.}
\begin{itemize}
    \item We isolate a weakly \textsf{NP}-complete restriction
    of bounded-supply EEAG with two agents, identical valuations, unit supplies,
    one initially endowed good, and no explicit budget bound.

    \item We give an $O(n^3)$ algorithm for bounded-supply
    single-type EEAG with $n$ agents and nonnegative integer per-copy values.
    The difference-constraint argument also applies whenever a repair action is
    summarized by an integer amount per recipient and each observer assigns a
    fixed additive marginal value to one unit.

    \item We prove that bounded-supply EEAG is \textsf{NP}-complete with
    exactly two additional item types, binary additive valuations, and no
    explicit budget bound.  Both supplies are positive, and the approval set
    of one type is contained in that of the other.  Together with the
    single-type algorithm, this establishes a sharp type-count dichotomy for
    binary valuations.
\end{itemize}

\subsection{Related Work}
\label{subsec:related_work}

\paragraph{Fair division of indivisible goods.}
Envy-freeness is one of the classical fairness notions in fair division. For
indivisible goods, exact envy-freeness may fail to exist even in elementary
instances, and deciding whether an envy-free allocation exists is
\textsf{NP}-complete in general~\cite{bouveret2008efficiency}. This has motivated the study of relaxations
such as envy-freeness up to one good (EF1)~\cite{Lipton2004,caragiannis2019unreasonable}, 
envy-freeness up to any good (EFX)~\cite{caragiannis2019unreasonable},
and maximin-share guarantees (MMS)~\cite{Budish2011,kurokawa2016can}. 
EF1 allocations are known to exist for additive
valuations and can be computed by simple procedures such as round-robin or
envy-cycle elimination. EFX is stronger and remains open in full generality,
although it is known to exist in several important special cases. 
Maximum Nash welfare allocations provide another influential approach, since they combine
strong fairness and efficiency guarantees in several settings~\cite{caragiannis2019unreasonable}, 
but exact optimization is computationally difficult.

\paragraph{Repairing unfair initial allocations.}
Several recent works study how to improve or repair an allocation that has
already been implemented. One approach is to remove goods from the allocation:
Boehmer et al.~\shortcite{BoehmerEtAl2024Donating} study whether envy can be
eliminated by donating, equivalently deleting, a bounded number of goods from an
initial allocation. Another approach is to allow limited sharing of resources:
Bredereck et al.~\shortcite{BredereckEtAl2023Sharing} investigate how an allocation
can be improved when some resources may be shared by pairs of agents. A third
approach is to modify the initial allocation by exchanges: Yuen et
al.~\shortcite{YuenEtAl2025Exchanging} study the complexity of reaching an EF1
allocation from an initially unfair allocation via a sequence of item exchanges.
These models are close in spirit to ours because the initial allocation is not
discarded; instead, the goal is to restore fairness through a restricted type of
intervention.

\paragraph{Subsidies, transfers, and adding goods.}
Another line of work restores envy-freeness by introducing a divisible resource,
typically money. Halpern and Shah~\shortcite{HalpernShah2019Subsidy} study when a
given allocation of indivisible goods can be made envy-free by paying subsidies
to agents and how to compute minimum payments. Brustle et
al.~\shortcite{BrustleEtAl2020OneDollar} show that, for additive valuations with
appropriately normalized item values, one dollar per agent is sufficient to
guarantee the existence of an envy-free allocation with subsidies. 
Aziz~\shortcite{Aziz2021Transfers} further studies envy-freeness and equitability with
monetary transfers. In contrast, Bentert et
al.~\shortcite{BentertEtAl25} propose eliminating envy by adding
indivisible goods rather than money. Our work follows this latter direction and
studies the computational boundary between hardness and tractability when the
additional goods have bounded supply.

\paragraph{Envy elimination by adding goods.}
The work most closely related to ours is Bentert et
al.~\shortcite{BentertEtAl25}.  Their unbounded-supply algorithm treats
non-proportional pairs constructively and reduces the remaining proportional
case to a totally unimodular integer program.  When some supplies are finite,
their Proposition~4 gives \textsf{NP}-hardness with three additional item types,
binary valuations, and no budget bound.  Their Proposition~5 gives
\textsf{W[1]}-hardness parameterized by the number of agents even for identical
valuations, while their Observation~6 gives fixed-parameter tractability for
the combined parameter consisting of the numbers of agents and additional item
types.  For a bounded extension size, their Theorem~7 gives
\textsf{W[1]}-hardness parameterized by the budget.  Our two-type theorem
sharpens their Proposition~4 by merging the type used for selected vertices
with the common type used for selected edges.  The proof replaces their
separate per-type counting with a coupled edge--endpoint inequality and thereby
resolves their explicit two-type question.  Together with our single-type
algorithm, it gives a polynomial-versus-\textsf{NP}-complete dichotomy for one
and two additional types under binary valuations.  Our other hardness result
is orthogonal in the sense that it fixes the number of agents at two but permits 
many one-copy types.

\paragraph{Organization.}
Sect.~\ref{sec:preliminaries} introduces the model and recalls the ESSE problem.
Sect.~\ref{sec:EEAG-TwoAgents} proves the two-agent hardness result, and
Sect.~\ref{sec:EEAG-OneItem} gives the single-type algorithm.
Sect.~\ref{sec:EEAG-TwoTypes} proves two-type hardness and derives the
type-count dichotomy. Finally, Sect.~\ref{sec:conclusion} concludes with open
directions.

\section{Preliminaries}
\label{sec:preliminaries}

\subsection{EEAG, Supply, and Budget}

Let $\mathcal A=\{a_1,a_2,\ldots,a_n\}$ be the agents, let $P$ be
the set of initially allocated goods, and let $R$ be the set of additional item
types.  An initial allocation $\sigma$ assigns pairwise disjoint bundles
$\sigma(a_i)\subseteq P$.  Each type $r\in R$ has supply
$s(r)\in\mathbb Z_{\geq 0}\cup\{\infty\}$, and each agent has a nonnegative
additive valuation $u_i$.  An extension is a vector
$\rho:\mathcal A\times R\to\mathbb Z_{\geq0}$ satisfying
$\sum_{i=1}^n\rho(a_i,r)\leq s(r)$ for every $r\in R$. 
Physically, $\rho(a_i,r)$ is the number of copies of
additional item type~$r$ assigned to agent~$a_i$. The initial bundles
$\sigma(a_i)$ remain unchanged, and the extension only augments them with
additional goods. The supply inequality ensures that the total number of
allocated copies of each type~$r$ does not exceed its available supply~$s(r)$.
Unused copies may remain unallocated. 
After extension $\rho$, agent $a_i$'s value for $a_j$'s bundle is 
\[
u_i(\sigma(a_j))+\sum_{r\in R}\rho(a_j,r)u_i(r).
\]
The extension is \emph{envy-resolving} if this value is at most
$a_i$'s value for her own extended bundle for every ordered pair~$(i,j)$.
A valuation $u_i$ is \emph{binary} if
$u_i(p)\in\{0,1\}$ for every initially allocated good $p\in P$
and $u_i(r)\in\{0,1\}$ for every additional item type $r\in R$.

A \emph{bounded-supply} instance has $s(r)<\infty$ for at
least one type.  This is distinct from imposing a \emph{budget} $L$ requiring
$\sum_{i,r}\rho(a_i,r)\leq L$.  Our hardness result has finite
unit supplies and no separate budget bound.  In a \emph{single-type} instance,
$R=\{g\}$. If $s(g)=K<\infty$, the supply constraint is simply
$\sum_{i=1}^n\rho(a_i,g)\leq K$.

\subsection{Equal Sum Subsets with an Enforced Element}

We use the following indexed formulation so that equal-valued
input elements remain distinct.

\begin{definition}[ESS with Enforced Element (ESSE)]
Given positive integers $q_1,q_2,\ldots,q_m$, with~$q_m$
distinguished, the \emph{Equal Sum Subsets with an Enforced Element} problem
asks whether there are disjoint index sets~$X,Y\subseteq[m]$ such that~$m\in X$ 
and~$\sum_{t\in X}q_t=\sum_{t\in Y}q_t$.
\end{definition}
\begin{theorem}[Cieliebak et al.~\shortcite{CieliebakEtAl08}, Theorem~2]
\label{thm:ESSE-NPC}
ESSE is \textsf{NP}-complete.
\end{theorem}
Cieliebak et al.~\shortcite{CieliebakEtAl08} also give a
pseudo-polynomial dynamic program running in $O(m^2Q)$ time, where
$Q=\sum_{t=1}^m q_t$. Hence, ESSE is weakly \textsf{NP}-complete.

\section{Two-Agent Hardness with Unit Supplies}
\label{sec:EEAG-TwoAgents}

In this section, we restrict EEAG to two agents, identical valuations, one
initially endowed good, and unit supply for every additional type.  Although the
reduction below is deliberately direct, it identifies hardness in a direction
not covered by fixing the number of additional types.

\begin{definition}[Restricted two-agent bounded-supply EEAG]
\label{def:EF-init}
There are two agents $a_1,a_2$ with the same nonnegative
additive valuation $u$.  Initially, $a_1$ owns one good $p$ with $u(p)=h$,
whereas $a_2$ owns nothing.  The additional goods are
$G=\{g_1,g_2,\ldots,g_\ell\}$, each available once, with
$u(g_t)=w_t$.  An extension is specified by disjoint index sets
$I_1,I_2\subseteq[\ell]$, while goods outside~$I_1\cup I_2$ may remain unallocated.
Because the valuations are identical, the extension is envy-free exactly when
\[
h+\sum_{t\in I_1}w_t=\sum_{t\in I_2}w_t.
\]
The decision question is whether such $I_1,I_2$ exist.
\end{definition}

\begin{theorem}
\label{thm:EF-init-NPC}
The restricted two-agent bounded-supply case in
Definition~\ref{def:EF-init} is weakly \textsf{NP}-complete, even when $h$ and
all $w_t$ are positive integers encoded in binary.
\end{theorem}

\begin{proof}
A certificate consists of $I_1$ and $I_2$.  Their disjointness
and the equality in Definition~\ref{def:EF-init} can be verified in polynomial
time, so the problem belongs to \textsf{NP}.

For hardness, take an arbitrary ESSE instance
$q_1,q_2,\ldots,q_m$ whose enforced element is~$q_m$.  Create the initial good $p$
with value $h:=q_m$ and, for every $t\in[m-1]$, one additional good $g_t$
with value $w_t:=q_t$.  All these values are common to both agents, and every
additional good has supply one.  The construction is polynomial in the binary
encoding of the ESSE instance.

Suppose first that disjoint index sets $X,Y\subseteq[m]$
witness ESSE, with $m\in X$.  Put $I_1:=X\setminus\{m\}$ and $I_2:=Y$.
Then 
\[
h+\sum_{t\in I_1}w_t
=q_m+\sum_{t\in X\setminus\{m\}}q_t
=\sum_{t\in Y}q_t
=\sum_{t\in I_2}w_t, 
\]
so the extension is envy-free.  Conversely, if disjoint
$I_1,I_2\subseteq[m-1]$ satisfy the displayed equality in
Definition~\ref{def:EF-init}, then $X:=I_1\cup\{m\}$ and $Y:=I_2$ are
disjoint, contain the enforced index in $X$, and have equal sums.  Hence the
two instances are equivalent, proving \textsf{NP}-hardness.

It remains to justify the qualifier ``weakly.''  Finally, we show that the equivalence
also works from any instance of Definition~\ref{def:EF-init} back to ESSE.
Use the list $w_1,w_2,\ldots,w_\ell,h$ and enforce its last element.  Therefore the
pseudo-polynomial ESSE algorithm of Cieliebak et al.~\cite{CieliebakEtAl08}
decides our instance in $O((\ell+1)^2(h+\sum_{t=1}^{\ell} w_t))$ time.  Together with
\textsf{NP}-completeness, this proves weak \textsf{NP}-completeness.
\end{proof}

\paragraph{Remark.} Our hardness result imposes no separate budget on the total number of
additional goods. Any subset of the $\ell$ available goods, including all of
them, may be allocated. The only supply restriction is that each additional
good is available once. This result is orthogonal to the type-count dichotomy
proved in Sect.~\ref{sec:EEAG-TwoTypes}, where the number of additional types
is fixed at two while the number of agents may grow. Here, the number of agents
is fixed at two while the number of unit-supply types may grow.

\section{Single-Type EEAG with Bounded Supply}
\label{sec:EEAG-OneItem}

In this section, we consider EEAG by adding extra copies of a \emph{single} good type. 
Let $\mathcal A=\{a_1,a_2,\ldots,a_n\}$ be the agents and let~$B_i$ be the initial bundle of~$a_i$.  
For every ordered pair~$(i,j)$, define 
\[
U_i^{\text{self}}:=u_i(B_i)\quad\text{ and }\quad
U_i^{(j)}:=u_i(B_j).
\]
We need no further structure on the values of the original
bundles. That is, the matrix consisting of these $n^2$ numbers is sufficient.  This
statement concerns only the original bundles.  For the added goods, we
explicitly assume additive and separable utilities as follows.  There is one
additional item type $g$ with finite supply $K$, and agent $a_i$ assigns value
$v_i\in\mathbb Z_{\geq0}$ to each copy.  Hence, if $a_j$ receives $x_j$
copies, then $a_i$ values $a_j$'s final bundle at
$U_i^{(j)}+v_i x_j$.  All input numbers are nonnegative integers encoded in
binary.

\begin{definition}[Bounded-supply single-type EEAG]
\label{def:single-type}
A vector $x\in\mathbb Z_{\geq0}^n$ is an envy-resolving
extension if $\sum_{i=1}^n x_i\leq K$ and, for every ordered pair
$(i,j)$, 
\[
U_i^{\text{self}}+v_i x_i \;\geq\; U_i^{(j)}+v_i x_j.
\]
\end{definition}

\paragraph{Difference constraint formulation.} 
Put $\Delta_{ij}:=U_i^{(j)}-U_i^{\text{self}}$.  If $v_i=0$, the
constraint for $(i,j)$ is feasible exactly when $\Delta_{ij}\leq0$.
If $v_i>0$, integrality of $x_i-x_j$ makes the same constraint equivalent to 
\begin{equation}
\label{eq:single-difference}
x_i-x_j\geq c_{ij},
\quad
c_{ij}:=\left\lceil\frac{\Delta_{ij}}{v_i}\right\rceil.
\end{equation}
Importantly,~\eqref{eq:single-difference} is retained even when
$\Delta_{ij}\leq 0$: assigning sufficiently many copies to~$a_j$ can create
new envy from~$a_i$ to~$a_j$.

\paragraph{Why a natural greedy heuristic can fail.}
Consider the rule that repeatedly assigns the next copy to an
agent whose current own-bundle utility
$U_i^{\text{self}}+v_i x_i$ is smallest.  Let there be two agents and
$K=2$, with
\begin{align*}
&U_1^{\text{self}}=5,\quad U_1^{(2)}=0,\quad v_1=1, \\
\quad
&U_2^{\text{self}}=8,\quad U_2^{(1)}=20,\quad v_2=10.
\end{align*}
The agents' own-bundle utilities evolve as 
$(5,8)\to(6,8)\to(7,8)$, so the heuristic assigns both copies to $a_1$.
The resulting vector $(2,0)$ is not envy-resolving: agent $a_2$ values her
own final bundle at~$8$, but values~$a_1$'s final bundle at
$20+10\cdot2=40$.  Nevertheless, the vector $(0,2)$ is envy-resolving,
because $a_1$ compares $5$ with~$0+1\cdot 2 = 2$, while~$a_2$ compares
$8+10\cdot2=28$ with~$20$.  Thus, locally balancing the agents'
own-bundle utilities can move the allocation in the wrong direction since 
assigning a copy to one agent changes how every observer evaluates that
agent's bundle.  This motivates treating all envy inequalities jointly
through difference constraints.

\subsection{Warm-up: Binary Per-Copy Values}
\label{subsec:binary-warmup}

We first specialize to $v_i\in\{0,1\}$.  This case is
subsumed by the general theorem in the next subsection, nevertheless, it is useful as 
a warm-up. It isolates the minimum-copy formulation and its integrality before we
turn to the more constructive shortest-path argument.  If $v_i=0$ and
$\Delta_{ij}>0$ for some~$j$, the instance is infeasible.  Otherwise, every
agent with $v_i=1$ contributes the constraints 
\begin{equation}
\label{eq:binary-difference}
x_i-x_j\geq\Delta_{ij}\quad\text{for every }j\in[n].
\end{equation}
As above, constraints with $\Delta_{ij}\leq 0$ must not be
discarded, since added copies can create envy that was absent initially.

\begin{lemma}
\label{lem:TU-difference}
Let $e_i\in\mathbb R^n$ be the $i$-th standard unit vector in which the $i$-th 
coordinate is 1 and other coordinates are~$0$. Let $A$ be any matrix whose rows 
are of the following two types: 
\begin{enumerate}
    \item rows of the form~$e_i-e_j$, corresponding to difference constraints $x_i-x_j\geq c_{ij}$, 
    and 
    \item rows of the form~$e_i$, corresponding to non-negativity constraints $x_i\geq 0$.
\end{enumerate}
Then $A$ is totally unimodular. That is, every square submatrix of~$A$ has determinant in~$\{-1, 0, 1\}$.
\end{lemma}

\begin{proof}
We prove the claim directly by induction on the order of a
square submatrix.  Let $B$ be any $k\times k$ submatrix of $A$.  The claim is
immediate for $k=1$, since every entry of $A$ belongs to
$\{-1,0,1\}$.

Suppose first that some row of~$B$ contains at most one
nonzero entry.  If that row is zero, then $\det(B)=0$.  Otherwise, its unique
nonzero entry is either~$1$ or~$-1$, and cofactor expansion along that row
expresses $\det(B)$, up to sign, as the determinant of a
$(k-1)\times(k-1)$ submatrix of~$A$.  By the induction hypothesis, this
determinant belongs to~$\{-1,0,1\}$.

It remains to consider the case in which every row of $B$
contains at least two nonzero entries.  Since every row of $A$ contains at
most two nonzero entries, every row of~$B$ must contain exactly two.
Consequently, $B$ contains no row arising from a unit vector~$e_i$, and each
of its rows arising from $e_i-e_j$ retains both its~$1$ and its~$-1$.
Therefore every row of~$B$ sums to zero, so
$B(1,1,\ldots,1)^{\top} = (0, 0, \ldots, 0)^{\top}$.  Hence $B$ is singular and $\det(B)=0$.
Thus, every square submatrix of~$A$ has determinant in
$\{-1,0,1\}$, proving that~$A$ is totally unimodular.
\end{proof}

\begin{theorem}
\label{thm:binary-case}
When $v_i\in\{0,1\}$ for every agent, bounded-supply
single-type EEAG is solvable in polynomial time.
\end{theorem}

\begin{proof}
After rejecting any pair with $v_i=0$ and
$\Delta_{ij}>0$, consider the linear program (LP)
\begin{equation}
\label{lp:binary-warmup}
\begin{array}{ll}
\text{minimize} & \displaystyle\sum_{i=1}^n x_i\\[1mm]
\text{subject to} & x_i-x_j\geq\Delta_{ij}
    \quad \forall(i,j)\text{ with }v_i=1,\\
& x_i\geq 0 \quad \forall i\in[n].
\end{array}
\end{equation}
Its constraints are exactly the remaining envy constraints
and nonnegativity of variables. The supply bound is deliberately omitted from the LP and
checked against its minimum objective value.  By
Lemma~\ref{lem:TU-difference}, the constraint matrix is totally unimodular.
All right-hand sides are integers, so whenever the LP is feasible it has an
integral optimal extreme point~$\mathbf{x}^* := (x^*_1,x^*_2,\ldots,x^*_n)$.

If the LP is infeasible, no envy-resolving extension exists.
If $\sum_i x_i^*\leq K$, assigning $x_i^*$ copies to $a_i$ is a feasible
extension within supply.  Conversely, every envy-resolving extension is an
integer feasible solution of~\eqref{lp:binary-warmup}. Hence, if
$\sum_{i=1}^n x_i^*>K$, no such extension can use at most~$K$ copies.  Linear
programming therefore decides the binary case and constructs a solution in
polynomial time.
\end{proof}

The warm-up highlights two points that remain important as follows. 
All ordered-pair constraints must be retained, and the supply cap should be
compared with a minimum-size feasible extension.  The same totally-unimodular structure 
remains after replacing~$\Delta_{ij}$ by~$c_{ij}$ for general~$v_i>0$. Nevertheless,
we use Bellman--Ford next because it directly returns the componentwise least
extension and makes the sign convention explicit.

\subsection{General Nonnegative Integer Per-Copy Values}
\label{subsec:general-values}

\begin{theorem}
\label{thm:one-item-type-general-value}
Bounded-supply single-type EEAG can be decided, and an
envy-resolving extension can be constructed when one exists, in polynomial-time.
\end{theorem}

\begin{proof}
First, we note that if some ordered pair has $v_i=0$ and $\Delta_{ij}>0$, no
extension can remove that envy, so we can simply reject for this case.  
Otherwise, the complete system
apart from the supply cap consists of~\eqref{eq:single-difference} for all
pairs with $v_i>0$, together with $x_i\geq0$.  These are lower-bound
difference constraints.

To put them in the standard shortest-path sign convention,
set $y_i:=-x_i$.  Then 
\[
x_i-x_j\geq c_{ij}\quad\text{ if and only if }\quad
y_i\leq y_j-c_{ij},
\]
while $x_i\geq 0$ becomes $y_i\leq 0$.  Construct a directed
graph~$H$ on vertices~$\{s,1,2,\ldots,n\}$ in which~$s$ is an auxiliary 
vertex as the \emph{source}.  For every constraint indexed by~$(i,j)$ with~$v_i>0$, 
add the edge $(j,i)$ of weight~$-c_{ij}$. For every~$i$, add the
edge~$(s,i)$ of weight~$0$.  Fixing $y_s=0$, an edge~$(p,q)$ of weight~$w$
represents the upper bound $y_q\leq y_p+w$.  This explains both the edge
orientation and the minus sign.

Because every vertex is reachable from~$s$, the constraint
system is feasible if and only if~$H$ has no negative-weight cycle.  
Indeed, if $H$ contains a negative-weight cycle, then summing the corresponding 
inequalities around the cycle gives an impossibility of the form~$y_p\leq y_p+C$ 
for some~$C < 0$. Hence no feasible vector~$\mathbf{y}$, and no feasible allocation 
vector~$\mathbf{x}$, exists.  
Conversely, if there is no negative cycle, let~$d_i$ be the
shortest-path distance from~$s$ to~$i$.  The shortest-path inequalities give
$d_i\leq d_j-c_{ij}$ for every edge~$(j,i)$ and~$d_i\leq 0$ for every direct
edge~$(s,i)$.  Therefore 
\[
x_i^*:=-d_i 
\]
is nonnegative and satisfies
$x_i^*-x_j^*\geq c_{ij}$ for every required ordered pair.

We next prove the key supply claim.  Let $\mathbf{x}$ be any feasible
vector for the lower-bound system and put $y_i=-x_i$.  Along every directed
path from~$s$ to~$i$, repeated application of the corresponding upper-bound
constraints shows that~$y_i$ is at most the path's weight.  It is therefore at
most the minimum such weight: $y_i\leq d_i$.  Consequently,
$x_i\geq-d_i=x_i^*$ for every~$i$.  Thus, $\mathbf{x}^*$ is the componentwise least
feasible extension and, in particular, minimizes $\sum_{i=1}^n x_i$.  An extension
within supply exists exactly when $\sum_{i=1}^n x_i^*\leq K$.

All $c_{ij}$ and all edge weights are integers, so the returned
vector~$\mathbf{x}^*$ is integral.  The graph has $n+1$ vertices and~$O(n^2)$ edges.
Bellman--Ford shortest-path algorithm takes $O(n^3)$ arithmetic operations~\cite{Ford56,Bellman58}. 
The numbers manipulated have polynomial bit length.  Construction and the final supply check take
$O(n^2)$ additional time.  Algorithm~\ref{alg:single-type-eeag-bf} below
summarizes the procedure.
\end{proof}

\subsection{Single-Type Algorithm}
\label{sec:single-type-algorithm}

\begin{algorithm}[htbp]
\caption{Single-Type EEAG via Bellman--Ford}
\label{alg:single-type-eeag-bf}
\small
\DontPrintSemicolon
\KwIn{Values $U_i^{\text{self}}$ and $U_i^{(j)}$ for all
$i,j\in[n]$, per-copy values $v_i\in\mathbb Z_{\geq 0}$, and supply~$K$.}
\KwOut{An envy-resolving vector $\mathbf{x}$, or \textsc{No}.}

\ForEach{ordered pair $(i,j)\in[n]\times[n]$}{
    $\Delta_{ij}\gets U_i^{(j)}-U_i^{\text{self}}$\;
    \If{$v_i=0$ and $\Delta_{ij}>0$}{
        \Return \textsc{No}\;
    }
}
Construct $H$ with vertex set $\{s,1,2,\ldots,n\}$\;
\ForEach{$i\in[n]$}{
    Add edge $(s,i)$ of weight~$0$\;
}
\ForEach{ordered pair $(i,j)$ with~$v_i>0$}{
    $c_{ij}\gets\lceil\Delta_{ij}/v_i\rceil$\;
    Add edge $(j,i)$ of weight $-c_{ij}$\;
}
Run Bellman--Ford from $s$ and let $d_i$ be the distance to~$i$\;
\If{$H$ contains a negative-weight cycle}{
    \Return \textsc{No}\;
}
\ForEach{$i\in[n]$}{
    $x_i^*\gets-d_i$\;
}
\eIf{$\sum_{i=1}^n x_i^*>K$}{
    \Return \textsc{No}\;
}{
    \Return $\mathbf{x}^*=(x_1^*,x_2^*,\ldots,x_n^*)$\;
}
\end{algorithm}

\paragraph{Scope of the argument.}
The algorithm uses the original allocation only through the bundle values
$U_i^{\text{self}}$ and~$U_i^{(j)}$. Hence, it does not require the
valuations over the original goods to be additive. By contrast, the utility
from the added copies must be additive. If agent~$a_j$ receives~$x_j$ copies
of~$g$, then agent~$a_i$ must assign them the additional value~$v_i x_j$.
This assumption is essential because it allows each envy constraint to be
written as
\[
v_i(x_i-x_j)\geq \Delta_{ij},
\]
and therefore as a difference constraint when~$v_i>0$.

\section{Two-Type Hardness and the Type-Count Dichotomy}
\label{sec:EEAG-TwoTypes}

We now close the type-count boundary for binary additive valuations.  Bentert
et al.~\shortcite{BentertEtAl25} prove the \textsf{NP}-hardness with three additional item
types and ask whether two types suffice.  The following theorem answers this
question affirmatively, even when both types have positive finite supply and
there is no separate budget bound.

\begin{theorem}
\label{thm:two-type-hardness}
Bounded-supply EEAG is \textsf{NP}-complete even under all of the following
restrictions: 
\begin{enumerate}
    \item there are exactly two additional item types both with finite
    positive supply;
    \item all valuations, for initial and additional goods, are binary and
    additive; and 
    \item there is no explicit budget bound on the size of the extension.
\end{enumerate}
Moreover, all numerical quantities in the reduction are
polynomially bounded.
\end{theorem}

\begin{proof}
For membership in \textsf{NP}, consider an arbitrary instance satisfying the
stated restrictions, and let $R=\{r,q\}$ denote its two additional item types 
with bounded supply $s(r)$, $s(q)$, respectively. 
A certificate specifies a nonnegative
integer count $\rho(a,t)$ for every agent $a\in\mathcal A$ and every
type $t\in R$. Since every feasible count satisfies
$0\leq\rho(a,t)\leq s(t)$, each count has binary length
$O(\log(s(t)+1))$, and the complete certificate has polynomial length.

For each $t\in R$, we first verify the supply constraint
\[
\sum_{a\in\mathcal A}\rho(a,t)\leq s(t).
\]
We then verify envy-freeness by checking, for every ordered pair of agents
$(a,a')$, that
\[
\begin{aligned}
u_a(\sigma(a))+\sum_{t\in R}\rho(a,t)u_a(t)
&\geq u_a(\sigma(a')) +\sum_{t\in R}\rho(a',t)u_a(t).
\end{aligned}
\]
There are two supply constraints and quadratically many envy inequalities,
and every calculation involves explicitly encoded integers of polynomial
bit length. Thus, feasibility and envy-freeness can be verified in
polynomial time.

For hardness, we reduce from \textsc{Clique} restricted to instances with
$3\leq\ell\leq |V|$, which remains \textsf{NP}-complete~\cite{Karp72}.  Let
$(G=(V,E),\ell)$ be such an instance, and put $M:=\binom{\ell}{2}$.
We first ensure that a supply defined below is positive.  If $|E|<M+1$, add
$M+1-|E|$ vertex-disjoint copies of $K_2$.  More precisely, for each copy,
introduce two fresh vertices joined by a single edge; add no edge between a
fresh vertex and an original vertex or between vertices belonging to distinct
added copies.  Call the resulting graph $G^+=(V^+,E^+)$.  Otherwise, set
$G^+=G$.  Since $\ell\geq3$, none of the added two-vertex components contains
or intersects an $\ell$-clique.  Hence $G^+$ contains an $\ell$-clique if and
only if $G$ does.  Relabel $G^+$ as $G=(V,E)$ for the remainder of the proof;
we now have $|E|\geq M+1$.

Create one vertex agent $a_v$ for each $v\in V$, one edge agent $a_e$ for
each $e\in E$, and one distinguished agent~$b$.  Construct the fixed initial
allocation from the following binary-valued goods:
\begin{itemize}
    \item agent $b$ receives a good $p_b$ approved by~$b$ and every edge
    agent, and by no vertex agent;
    \item each vertex agent $a_v$ receives a good~$p_v$ approved only by
    $a_v$; and
    \item for every edge $e=\{u,v\}$, edge agent~$a_e$ receives a good~$p_e$
    approved exactly by~$a_u$ and~$a_v$.
\end{itemize}
In this initial allocation, the only envy is from edge agents to~$b$. Indeed, 
every~$a_e$ values its own bundle at zero and $b$'s bundle at one. We observe 
that~$b$ values only its own good, each vertex agent values its own bundle and every
incident edge-agent bundle equally at one, and edge agents value all bundles
other than $b$'s at zero.  Thus any envy-resolving extension must give every
edge agent at least one approved additional good.

There are two additional item types, denoted by~$r$ and~$q$.  Their approval
patterns and supplies are listed specifically below:
\begin{center}
\begin{tabular}{@{}lcc@{}}
\toprule
Agent class & $r$ & $q$ \\
\midrule
vertex agent $a_v$ & $1$ & $0$ \\
edge agent $a_e$ & $1$ & $1$ \\
distinguished agent $b$ & $0$ & $0$ \\
\bottomrule
\end{tabular}
\qquad
$\begin{aligned}
s(r)&=M+\ell,\\
s(q)&=|E|-M.
\end{aligned}$
\end{center}
By the preceding padding step, $|E|\geq M+1$, and hence
$s(q)=|E|-M\geq1$.  Moreover, $s(r)=M+\ell>0$.  Thus both additional item
types have strictly positive, polynomially bounded supply.  Notice also that
every agent approving $q$ approves~$r$.

Suppose first that $K\subseteq V$ is a clique of size~$\ell$, and let
\[
E_K:=\bigl\{\{u,v\}\in E:u,v\in K\bigr\}
\]
be its set of~$M$ internal edges.  Allocate one copy of~$r$ to~$a_e$ for every
$e\in E_K$, one copy of~$q$ to~$a_e$ for every~$e\in E\setminus E_K$, and one
copy of~$r$ to~$a_v$ for every~$v\in K$.  This uses exactly~$M+\ell$ copies
of~$r$ and $|E|-M$ copies of~$q$.

We check envy-freeness by agent class.  Every edge agent values its own final
bundle at one, $b$'s bundle at one, every other edge-agent bundle at one, and
every vertex-agent bundle at most one.  Hence no edge agent envies.  
Agent~$b$ values only its own initial good and therefore does not envy.  A vertex
agent~$a_v$ values its own initial bundle at one and values the initial bundle
of an edge agent~$a_e$ at one precisely when~$v$ is incident to~$e$.  If
$e\in E_K$ is incident to~$v$, then~$v\in K$. Both $a_v$'s own bundle and
$a_e$'s bundle have value two for~$a_v$.  Every other bundle has value at most
$a_v$'s own value.  Thus the extension is envy-resolving. 

Conversely, suppose that an envy-resolving extension~$\rho$ exists.  Every
edge agent must receive at least one copy of $r$ or~$q$.  Let
\[
E_r:=\{e\in E:\rho(a_e,r)\geq 1\},
\quad k:=|E_r|.
\]
Only $|E|-M$ copies of~$q$ are available, that is, $\bigl|\{e\in E : \rho(a_e,q)\geq 1\}\bigr|
  \leq s(q) = |E|-M$.  If fewer than $M$ edge agents
received $r$, more than $|E|-M$ edge agents would require a copy of~$q$.
Consequently,
\begin{equation}
\label{eq:two-type-k-lower}
k\geq M.
\end{equation}

Let $W\subseteq V$ be the set of vertices incident to at least one edge in~$E_r$.  
For every $v\in W$, choose an incident edge~$e\in E_r$.  Agent~$a_v$
values~$a_e$'s initial good at one and the copy of~$r$ held by~$a_e$ at one,
so it values~$a_e$'s final bundle at least two.  Since~$a_v$ values its own
initial bundle at one and does not value~$q$, avoiding envy requires
$\rho(a_v,r)\geq 1$.  Thus at least~$k$ copies of~$r$ go to the edge agents 
in~$E_r$, and at least~$|W|$ further copies go to their incident vertex agents.
The supply of~$r$ gives
\begin{equation}
\label{eq:two-type-supply-count}
k+|W|\leq M+\ell.
\end{equation}
By~\eqref{eq:two-type-k-lower} and~\eqref{eq:two-type-supply-count},
$|W|\leq\ell$.  By the definition of~$W$, every edge $\{u,v\}\in E_r$
satisfies~$u,v\in W$.  Since $G$ is simple and~$k=|E_r|$, it follows that
\[
M\leq k\leq\binom{|W|}{2}\leq\binom{\ell}{2}=M.
\]
All inequalities become equalities.  Hence $k=M$, $|W|=\ell$, and~$E_r$ consists
of all~$M$ possible edges on~$W$.  Therefore, $W$ is an $\ell$-clique of the
padded graph and hence of the original graph.  This proves the reduction.

The constructed supplies are~$O(|V|^2+|E|)$, while every valuation is binary.
Thus the hardness also persists when all numerical quantities are encoded in
unary.
\end{proof}

\paragraph{Why the merge works.}
The three-type reduction of Bentert et al.~\shortcite{BentertEtAl25} uses a
common type for selected edges, an edge-only type for unselected edges, and a
vertex-only type for selected vertices.  We merge the selected-edge and
selected-vertex types into~$r$ and add their supplies.  An edge agent then also
values copies of $r$ assigned to selected vertices, but this creates no new
envy, since every edge agent receives one approved additional good and values both
bundles at one.  In the reverse direction, separate per-type counting is
replaced by
\[
|E_r|+|W|\leq\binom{\ell}{2}+\ell.
\]
The supply bound $s(q)=|E|-M$ implies that at most $|E|-M$ of the $|E|$ edge agents 
can receive a copy of~$q$. Since each edge agent must receive~$r$ or~$q$ to resolve 
envy, it forces that at least $|E|-(|E|-M) = M=\binom{\ell}{2}$ edge agents to 
receive~$r$, and hence
$|E_r|\geq\binom{\ell}{2}$. However, a simple graph on at most~$\ell$ incident 
vertices has at most $\binom{\ell}{2}$ edges.  Equality
therefore forces a clique.

Combining Theorem~\ref{thm:two-type-hardness} with the single-type algorithm
gives the promised sharp boundary.

\begin{corollary}
\label{cor:type-count-dichotomy}
For bounded-supply EEAG with binary additive valuations and no explicit budget
bound, the problem is polynomial-time solvable with one additional item type
and \textsf{NP}-complete with exactly two additional item types.
\end{corollary}

\begin{proof}
The one-type case follows from
Theorem~\ref{thm:one-item-type-general-value}, which permits arbitrary
nonnegative integer per-copy values.  The two-type case follows from
Theorem~\ref{thm:two-type-hardness}.
\end{proof}

\paragraph{Remark.} Corollary~\ref{cor:type-count-dichotomy} does not subsume the two-agent
hardness result of Sect.~\ref{sec:EEAG-TwoAgents}.  The dichotomy fixes the
number of types and allows the number of agents to grow, whereas
Theorem~\ref{thm:EF-init-NPC} fixes the number of agents at two and allows the
number of unit-supply types to grow.

\section{Conclusion}
\label{sec:conclusion}

We established a sharp type-count dichotomy for bounded-supply EEAG under
binary additive valuations and no explicit budget bound.  With one additional
item type, the envy inequalities admit a componentwise least solution
computable by shortest paths, so a finite supply cap can be checked in
polynomial time for any number of agents.  With exactly two additional item
types, the problem is \textsf{NP}-complete, even when both supplies are positive
and the approval set of one type is contained in that of the other.  This
closes the two-type question posed by Bentert et al.~\shortcite{BentertEtAl25}.
Orthogonally, exact balancing yields weak \textsf{NP}-completeness even for two
agents with identical valuations when the number of one-copy types may grow.

These results leave several refinements of the two-agent case to investigate.
In particular, it would be interesting to determine whether the weak
\textsf{NP}-completeness result can be strengthened under natural valuation
restrictions, or whether broader two-agent settings admit pseudo-polynomial
algorithms. Another direction is to identify broader valuation or intervention
structures for which the envy inequalities retain a tractable
difference-constraint formulation.

\bibliographystyle{named}
\bibliography{EEAG}

\end{document}